\documentclass[12pt,a4paper]{article}

\usepackage[cp 1251]{inputenc}
\usepackage{amssymb,amsmath,amsthm,euscript}
\usepackage{hyperref}
\usepackage[dvips]{graphicx}

\def\Hinf{\mbox{${\mathcal{H}_\infty}$}}
\def\H2{\mbox{${\mathcal{H}_2}$}}
\def\calH{\mathcal{H}}

\def\R{\mathbb{R}}

\def\D{\mathcal{D}}
\def\L{\mathbb{L}}
\def\G{\mathbb{G}}

\def\AA{\mathbf{A}}
\def\MA{\overline{\bf A}}

\def\E{\mathbf{E}}
\def\e{\mathrm{e}}

\def\h{\mathbf{h}}

\def\sn{\mbox{$| \! | \! |$}}
\def\ds{\displaystyle}

\def\ol{\overline}
\def\wh{\widehat}

\def\whR{\wh{R}}
\def\whP{\wh{P}}

\def\olA{\ol{A}}
\def\cov{{\bf cov}}

\newcommand{\esssup}{\mathop{\mathrm{ess\,sup}}}

\newcommand{\tr}{\mathop{\mathrm{tr}}\nolimits}

\newcommand{\T}{\mathrm{T}}
\newcommand{\mes}{\mathop{\mathrm{mes\,}}\nolimits}

\newtheorem{theorem}{Theorem}
\newtheorem{lemma}{Lemma}

\newtheorem{remark}{Remark}

\def\be#1{\begin{equation}\label{#1}}
\def\ee{\end{equation}}

\begin{document}
\title{A Convex Formulation of Strict Anisotropic Norm Bounded Real Lemma\footnote{This work was supported
    by the Russian Foundation for Basic Research (grant 11-08-00714-a) and Program for Fundamental Research No.~15
of EEMCP Division of Russian Academy of Sciences.}}
\author{Michael M.~Tchaikovsky$^\dag$\qquad Alexander P.
Kurdyukov$^\dag$\qquad Victor N. Timin\footnote{The authors are
with Institute of Control Sciences of Russian Academy of Sciences,
Moscow, Russia, 117997, Profsoyuznaya 65, fax:
+7\,495\,334\,93\,40, e-mails: mmtchaikovsky@hotmail.com,
akurd@ipu.ru,\newline timin.victor@rambler.ru.}}
\date{\today}
\maketitle

\begin{abstract}
This paper is aimed at extending the $\Hinf$ Bounded Real Lemma to
stochastic systems under random disturbances with imprecisely
known probability distributions. The statistical uncertainty is
measured in entropy theoretic terms using the mean anisotropy
functional. The disturbance attenuation capabilities of the system
are quantified by the anisotropic norm which is a stochastic
counterpart of the $\Hinf$ norm. A state-space sufficient
criterion for the anisotropic norm of a linear discrete time
invariant system to be bounded by a given threshold value is
derived. The resulting Strict Anisotropic Norm Bounded Real Lemma
involves an inequality on the determinant of a positive definite
matrix and a linear matrix inequality. It is shown that slight
reformulation of these conditions allows the anisotropic norm of a
system to be efficiently computed via convex optimization.
\end{abstract}



\begin{flushright}
\emph{Dedicated to the blessed memory of our comrade and colleague
Eugene Maximov.}
\end{flushright}

\section{Introduction}

The anisotropy of a random vector and the anisotropic norm of a
system are the main concepts of the anisotropy-based theory of
robust stochastic control originally developed by I.G.~Vladimirov
and presented in~\cite{SVK_1994}--\cite{VKS_1996_1}.

The anisotropy functional considered there is an entropy theoretic
measure of the deviation of a probability distribution in
Euclidean space from Gaussian distributions with zero mean and
scalar covariance matrices. The mean anisotropy of a stationary
random sequence is defined as the anisotropy production rate per
time step for long segments of the sequence. In application to
random disturbances, the mean anisotropy describes the amount of
statistical uncertainty which is understood as the discrepancy
between the imprecisely known actual noise distribution and the
family of nominal models which consider the disturbance to be a
Gaussian white noise sequence with a scalar covariance matrix.

Another fundamental concept of I.G.\,Vladimirov's theory is the
$a$-anisotropic norm of a linear discrete time invariant (LDTI)
system which quantifies the disturbance attenuation capabilities
by the largest ratio of the power norm of the system output to
that of the input provided that the mean anisotropy of the input
disturbance does not exceed a given nonnegative parameter $a$.

In the context of robust stochastic control design aimed at
suppressing the potentially harmful effects of statistical
uncertainty, the anisotropy-based approach offers an important
alternative to those control design procedures that rely upon a
specific probability law of the disturbance and the assumption
that it is known precisely.

Minimization of the anisotropic norm of the closed-loop system as
a performance criterion leads to internally stabilizing dynamic
output feedback controllers that are less conservative than the
$\Hinf$ controllers and more efficient for attenuating the
correlated disturbances than the $\H2$ (LQG) controllers. A
state-space solution to the anisotropic optimal control problem
derived by I.G.\,Vladimirov in~\cite{VKS_1996_2} results in a
unique full-order estimator-based controller and involves the
solution of three cross-coupled algebraic Riccati equations, an
algebraic Lyapunov equation and a mean anisotropy equation on the
determinant  of a related matrix. Solving this complex system of
equations  requires application of a specially developed
homotopy-based numerical algorithm~\cite{DKSV_1997_report}.

The anisotropic suboptimal controller design is a natural
extension of this approach. Instead of minimizing   the
anisotropic norm of the closed-loop system, a suboptimal
controller is only required to keep it below a given threshold
value. Rather than resulting in a unique controller, the
suboptimal design yields a family of controllers, thus providing
freedom to impose some additional performance specifications on
the closed-loop system.

The ani\-so\-tro\-pic  suboptimal control design requires a
state-space criterion for verifying if the ani\-so\-tro\-pic norm
of a system does not exceed a given value. The Anisotropic Norm
Bounded Real Lemma (ANBRL) as a stochastic counterpart of the
$\Hinf$ Bounded Real Lemma for LDTI systems under statistically
uncertain stationary Gaussian random disturbances with bounded
mean anisotropy was presented in~\cite{KMT_2010}. The resulting
criterion has the form of an inequality on the determinant of a
matrix associated with an algebraic Riccati equation which depends
on a scalar parameter. A similar criterion for linear discrete
time varying systems involving a time-dependent inequality and
difference Riccati equation can be found in~\cite{MKV_2011}. This
paper aims at improving numerical tractability of ANBRL by
representing the criterion as a convex optimization problem. These
results are applied in~\cite{TKT_2011_asynCP} to design of the
suboptimal anisotropic controllers by means of convex optimization
and semidefinite programming.

The paper is organized as follows. Section~\ref{sect:background}
provides the minimum necessary background on the anisotropy of
signals and anisotropic norm of systems. Section~\ref{sect:main
result} establishes the Strict Anisotropic Norm Bounded Real Lemma
(SANBRL) which constitutes the main result of the paper. In
Subsection~\ref{sect:aninorm computing} we slightly reformulate
the SANBRL for efficient computation of the anisotropic norm of a
system by convex optimization. Subsection~\ref{sect:limiting
cases} considers $\H2$ and $\Hinf$ norms as two limiting cases of
the anisotropic norm. It is shown that in these cases the SANBRL
conditions transform to the well-known criteria for $\H2$ and
$\Hinf$ norms, respectively. Section~\ref{sect:numerical example}
presents benchmark results to compare the novel computational
algorithm with an earlier approach which employs a homotopy-based
algorithm for solving a system of cross-coupled nonlinear matrix
algebraic equations developed by
I.G.\,Vladimirov~\cite{DKSV_1997_report}. Concluding remarks are
given in Section~\ref{sect:conclusion}.

\subsection{Notation}

The set of reals is denoted by $\R,$ the set of real $(n\times m)$
matrices is denoted by $\R^{n\times m}.$ For a complex matrix $M =
[m_{ij}]$, $M^\ast$ denotes the Hermitian conjugate of the matrix:
$M^\ast:= [m^\ast_{ji}].$ For a real matrix $M = [m_{ij}]$, $M^\T$
denotes the transpose of the matrix: $M^\T := [m_{ji}].$ For real
symmetric matrices, $M\succ N$ ($M\succcurlyeq N$) stands for
positive
 definiteness (semidefiniteness) of $M-N$. The trace of a square matrix $M
=[m_{ij}]$ is denoted by $\tr{M}:=\sum_k{m_{kk}}.$ The spectral
radius of a matrix $M$ is denoted by
$\rho(M):=\max_k|\lambda_k(M)|,$ where $\lambda_k(M)$ is $k$-th
eigenvalue of the matrix $M.$ The maximum singular value of a
complex matrix $M$ is denoted by
$\ol{\sigma}(M):=\sqrt{\lambda_{\max}(M^\ast M)}.$ $I_n$ denotes a
$(n\times n)$ identity matrix, $0_{n\times m}$ denotes a zero
$(n\times m)$ matrix. The dimensions of zero matrices, where they
can be understood from the context, will be omitted for the sake
of brevity.

The angular boundary value of a transfer function $F(z)$ analytic
in the unit disc of the complex plane $|z|<1$ is denoted by
$$
{\widehat F}(\omega) :=\lim_{r\rightarrow 1-}{F(r\e^{i\omega})}.
$$
$\mathcal{H}_2^{p\times m}$ denotes the Hardy space of ($p\times
m$)-matrix-valued transfer functions $F(z)$ of a complex variable
$z$ which are analytic in the unit disc $|z|<1$ and have bounded
$\H2$ norm
$$
\|F\|_2 := \left(
\frac{1}{2\pi}\int\limits_{-\pi}^{\pi}{\tr(\wh{F}(\omega)\wh{F}^\ast(\omega))d\omega}
\right)^{1/2}.
$$
$\mathcal{H}_\infty^{p\times m}$ denotes the Hardy space of
($p\times m$)-matrix-valued transfer functions $F(z)$ of a complex
variable $z$ which are analytic in the unit disc $|z|<1$ and have
bounded $\Hinf$ norm
$$
\|F\|_\infty := \sup_{|z|\geqslant1}\ol{\sigma}(F(z)) =
\esssup_{-\pi\leqslant\omega\leqslant\pi}\ol{\sigma}(\wh{F}(\omega)).
$$

\section{Basic concepts of anisotropy-based robust performance
analysis}\label{sect:background}

For completeness of exposition, we provide the  minimum necessary
background material on the anisotropy of signals and anisotropic
norm of systems. Detailed information on the anisotropy-based
robust performance analysis developed originally by
I.G.\,Vladimitov~\cite{VKS_1995,VKS_1996_1} can be also found
in~\cite{DVKS_2001,VDK_2006}.

Let $\mathbb{L}_2^m$ denote the class of square integrable
$\R^m$-valued random vectors distributed absolutely continuously
with respect to the $m$-dimensional Lebesgue measure $\mes_m$. For
any $W\in\mathbb{L}_2^m$ with PDF $f\colon\R^m\to\R_+$,  the
\textit{anisotropy} $\AA(W)$  is defined in~\cite{VDK_2006} as the
minimal value of the relative entropy
$\mathbf{D}(f\|p_{m,\lambda})$ with respect to the Gaussian
distributions $p_{m,\lambda}$ in $\R^m$ with zero mean and scalar
covariance matrices $\lambda I_m$:
\begin{equation}
\label{AA}
    \AA(W)
    :=
    \min_{\lambda>0}
    \mathbf{D}(f\|p_{m,\lambda})
    =
    \frac{m}{2}
    \ln
    \left(
        \frac{2\pi\e}{m}
        \E|W|^2
    \right)
    -
    \h(W),
\end{equation}
where $\E$ denotes the expectation, $\h(W)$ denotes the
differential entropy  of $W$ with respect to
$\mes_m$~\cite{CT_1991}. It is shown in~\cite{VDK_2006} that the
minimum in~(\ref{AA}) is achieved at $\lambda = \E|W|^2/m$.

Let $W:= (w_k)_{-\infty< k<+\infty}$ be a stationary sequence of
vectors $w_k\in\L_2^m$ interpreted as a discrete-time random
signal. Assemble the elements of $W$ associated with a time
interval $[s,t]$ into a random vector
\begin{equation}
    W_{s:t}
    :=\left[
    \begin{array}{c}
        w_s\\
        \vdots\\
        w_t
    \end{array}\right].
\end{equation}
It is assumed that $W_{0:N}$ is distributed absolutely
continuously for every $N\geqslant 0$. The \textit{mean
anisotropy} of the sequence $W$ is defined in~\cite{VDK_2006} as
the anisotropy production rate per time step by
\begin{equation}
\label{overA0}
    \MA(W)
    :=
    \lim_{N\to+\infty}
    \frac
    {\AA(W_{0:N})}
    {N}.
\end{equation}
Let  $\G^m(\mu,\Sigma)$ denote the class of $\R^m$-valued Gaussian
random vectors with mean $\E w_k = \mu$  and nonsingular
covariance matrix $\cov(w_k):= \E(w_k-\mu)(w_k-\mu)^\T = \Sigma$.
Let $V:= (v_k)_{-\infty < k< +\infty}$ be a sequence of
independent random vectors $v_k\in\G^m(0,I_m)$, i.e. an
$m$-dimensional Gaussian white noise sequence. Suppose $W=GV$ is
produced from $V$ by a stable shaping filter with transfer
function $G(z)\in\calH_2^{m \times m}$. Then the spectral density
of $W$ is given by
\begin{equation}
\label{S}
    S(\omega)
    :=
    \wh{G}(\omega)
    \wh{G}(\omega)^*,
    \quad
    -\pi
    \leqslant
    \omega
    <
    \pi,
\end{equation}
where $\widehat{G}(\omega) := \lim_{r \to 1-} G(r{\rm
e}^{i\omega})$ is the boundary value of the transfer function
$G(z)$. It is shown in~\cite{VKS_1996_1,DVKS_2001} that the mean
anisotropy~(\ref{overA0}) can be computed in terms of the spectral
density~(\ref{S}) and the associated $\calH_2$ norm of the shaping
filter $G$ as
\begin{equation}
\label{meananiso}
    \MA(W)
    =
    -\frac{1}{4\pi}\,
    \int_{-\pi}^{\pi}
    \ln\det
        \frac
        {m S(\omega)}
        {\|G\|_2^2}
    d\omega.
\end{equation}
Since the probability law of the sequence $W$ is completely
determined by the shaping filter $G$ or by the spectral density
$S$, the alternative notations $\MA(G)$ and $\MA(S)$ are also used
instead of $\MA(W)$.

The mean anisotropy functional~(\ref{meananiso}) is always
nonnegative. It takes a finite value if the shaping filter $G$ is
of full rank, otherwise, $\MA(G) =
+\infty$~\cite{VKS_1996_1,DVKS_2001}. The equality $\MA(G)=0$
holds true if and only if $G$ is an all-pass system up to a
nonzero constant factor. In this case, spectral density~(\ref{S})
is described by $
    S(\omega)
    =
    \lambda I_m$,
$    -\pi
    \leqslant
    \omega
    <\pi
$, for some $\lambda >0$, so that  $W$ is a Gaussian white noise
sequence with zero mean and a scalar covariance matrix.

Let $F\in\mathcal{H}_\infty^{p\times m}$  be a LDTI system with an
$m$-dimensional input $W$ and a $p$-dimensional output $Z = FW$.
Let the random input sequence be given by $W=GV,$
 where, as before, $V$ is an $m$-dimensional Gaussian white noise sequence. Denote by
\begin{equation}
\label{calG}
    \mathcal{G}_a
    :=
    \left\{
        G
        \in
        \mathcal{H}_2^{m\times m}:
        \MA(G) \leqslant a
    \right\}
\end{equation}
the set of  shaping filters $G$ that produce Gaussian random
sequences $W$  with mean anisotropy (\ref{meananiso}) bounded by a
given parameter $a\geqslant 0$.

The $a$-\textit{anisotropic norm} of the system $F$ is defined by
I.G.\,Vladimirov as~\cite{VKS_1996_1,DVKS_2001}
\begin{equation}
\label{anorm}
    \sn F \sn_a
    :=
    \sup_{G \in \mathcal{G}_a}
    \frac
    {\|FG\|_2}
    {\|G\|_2}.
\end{equation}
It is shown in~\cite{VKS_1996_1} that the $a$-anisotropic norm of
a given system $F\in\mathcal{H}_\infty^{p\times m}$ is a
nondecreasing continuous function of the mean anisotropy level $a$
which satisfies
\begin{equation}
\label{limits}
    \frac{1}{\sqrt{m}}
    \|F\|_2
    =
    \sn F \sn_0
    \leqslant
    \lim_{a\to+\infty}
    \sn F \sn_a
    =
    \|F\|_\infty.
\end{equation}
These relations show that the $\H2$ and $\Hinf$ norms are the
limiting cases of the $a$-anisotropic norm  as $a\to 0, +\infty$,
respectively.

\section{Strict anisotropic norm bounded real lemma}\label{sect:main result}

Let  $F \in \mathcal{H}_{\infty}^{p \times m}$ be a LDTI system
with an $m$-dimensional input $W$, $n$-dimensional internal state
$X$ and $p$-dimensional output $Z$ governed by
\begin{equation}
\label{eq:plant}
    \begin{bmatrix}
        x_{k+1}\\
        z_k
    \end{bmatrix}
    =
    \begin{bmatrix}
        A & B\\
        C & D
    \end{bmatrix}
    \begin{bmatrix}
        x_k\\
        w_k
    \end{bmatrix},
\end{equation}
where $A$, $B$, $C$, $D$ are appropriately dimensioned real
matrices, and $A$ is stable (its spectral radius satisfies
$\rho(A)<1$). Suppose the input sequence $W$ is a stationary
Gaussian random sequence whose mean anisotropy does not exceed
$a\geqslant 0$, i.e. $W$ is produced from the $m$-dimensional
Gaussian white noise $V$ with zero mean and identity covariance
matrix by an unknown shaping filter $G$ which belongs to the
family $\mathcal{G}_a$ defined by~(\ref{calG}).

\subsection{Main result: a convex formulation}

The theorem below (SANBRL) provides a state-space criterion for
the anisotropic norm of the system~(\ref{eq:plant}) to be strictly
bounded by a given threshold $\gamma$.

\begin{theorem}\label{theorem:SANBRL det} Let $F \in \mathcal{H}_{\infty}^{p \times m}$ be a
system with the state-space realization (\ref{eq:plant}), where
$\rho(A)<1$. Then  its $a$-anisotropic norm (\ref{anorm}) is
strictly bounded by a given threshold $\gamma>0$, i.e.
\begin{equation}\label{eq:bounded aninorm ineq}
\sn F \sn_a
    <
\gamma
\end{equation}
if there exists $q\in(0, \min(\gamma^{-2}, \|F\|_\infty^{-2}))$
such that the inequality
\begin{equation}
    \label{determinant inequality}
    -(\det{(I_m-B^\T RB-qD^\T D)})^{1/m}
    < -(1-q\gamma^2)\e^{2a/m}
\end{equation}
holds true  for a real $(n\times n)$-matrix $R=R^\T\succ 0$
satisfying the linear matrix inequality
\begin{equation}\label{eq:aninorm LMI 2x2}
\left[
\begin{array}{cc}
A^\T RA-R & A^\T RB\\
B^\T RA & B^\T RB-I_m
\end{array}
\right]+q\left[
\begin{array}{c}
C^\T\\D^\T
\end{array}
\right]\left[
\begin{array}{cc}
C & D
\end{array} \right] \prec 0.
\end{equation}
\end{theorem}
\begin{remark}
Note that the constraints described by the
inequalities~(\ref{determinant inequality}) and (\ref{eq:aninorm
LMI 2x2}) are convex with respect to both variables $q$ and $R$.
Indeed, the function $-(\det(\cdot))^{1/m}$ of a positive definite
($m\times m$)-matrix on the left-hand side of~(\ref{determinant
inequality}) is convex~\cite{NN_1994,BTN_2000}.
\end{remark}
Before to proceed to proving the theorem, let us recall a
nonstrict formulation of ANBRL presented in~\cite{KMT_2010}.
\begin{lemma} \textbf{\emph{\cite{KMT_2010}}}\label{lemma:nonstrict ANBRL}
Let the assumptions of Theorem~\ref{theorem:SANBRL det} be
satisfied. Then
\begin{equation}\label{eq:nonstrict bounded aninorm ineq}
\sn F \sn_a
    \leqslant
\gamma
\end{equation}
if and only if there exists $q\in[0, \min(\gamma^{-2},
\|F\|_\infty^{-2}))$ such that the inequality
\begin{equation}
    \label{low_aniso}
    -\frac{1}{2}
    \ln\det
    (
        (1-q\gamma^2)
        \Sigma
    )
    \geqslant
    a
\end{equation}
is satisfied  for the matrix $\Sigma$ associated  with the
stabilizing ($\rho(A+BL)<1$) solution $\whR=\whR^\T\succcurlyeq 0$
of the algebraic Riccati equation
\begin{eqnarray}
    \label{Ric_R}
    \whR
    & = &
    A^\T \whR A + qC^\T C + L^\T \Sigma^{-1} L, \\
    \label{Ric_L}
    L
    & :=
    &
    \Sigma ( B^\T \whR A + qD^\T C ), \\
    \label{Ric_Sigma}
    \Sigma
    & := &
    ( I_m - B^\T \whR B - qD^\T D )^{-1}.
\end{eqnarray}
\end{lemma}
\begin{remark}\label{remark:nonstrict ANBRL}
Note that the matrix $\Sigma$ defined by~(\ref{Ric_Sigma}) is
positive definite if and only if $q<\|F\|_\infty^{-2}$. For any
such $q$, the left-hand side of the inequality
\begin{equation*}\label{eq:logdetSigma inequality}
-\ln\det{\Sigma} \geqslant m\ln{(1-q\gamma^2)} + 2a
\end{equation*}
equivalent to~(\ref{low_aniso}) is nonpositive since
$\Sigma\succcurlyeq I_m$. Therefore, any $q$
satisfying~(\ref{low_aniso}) must also satisfy
\begin{equation}\label{eq:q localization inequalities}
\gamma^{-2}(1-\e^{-2a/m}) \leqslant q < \gamma^{-2}.
\end{equation}
For every admissible value of $q$, the stabilizing solution $\whR$
of the Riccati equation (\ref{Ric_R})--(\ref{Ric_Sigma}) is
unique, so that there is a well-defined map $q\mapsto \whR_q$. The
set of those values of $q$ for which the pair $(q,\whR_q)$
satisfies the inequality (\ref{low_aniso}), form an interval
$[q_*,q^*]$ whose endpoints, for a given system $F$, are functions
of $a$ and $\gamma$. This interval becomes a singleton $q_*=q^*$
if and only if $\gamma = \sn F\sn_a$. For that reason, it is not
hard to derive the necessary and sufficient conditions for the
inequality in~(\ref{eq:nonstrict bounded aninorm ineq}) to be
strict. In this case the nonstrict inequality in~(\ref{low_aniso})
becomes the strict one resulting in similar modification
of~(\ref{eq:q localization inequalities}).
\end{remark}

To prove the main result, first we will need the following
assertion:
\begin{lemma}\label{lemma:from ineq to eq}
Let $F \in \mathcal{H}_{\infty}^{p \times m}$ be a system with the
state-space realization (\ref{eq:plant}), where $\rho(A)<1$, and
let the real positive values $\gamma$ and $a$ be given. Suppose
that there exist a real $(n\times n)$-matrix $R=R^\T\succ 0$ and
scalar value $q\in(0, \min(\gamma^{-2}, \|F\|_\infty^{-2}))$ such
that
\begin{multline}\label{eq:aninorm Ric ineq}
A^\T RA-R +qC^\T C+(A^\T RB+qC^\T D)(I_m-B^\T RB-qD^\T
D)^{-1}(B^\T RA+qD^\T C)\prec 0,
\end{multline}
\begin{equation}\label{eq:Sigma positive}
I_m-B^\T RB-qD^\T D \succ 0,
\end{equation}
and
\begin{equation}\label{eq:low ani ineq}
\ln\det{(I_m-B^\T RB-qD^\T D)}
> m\ln{(1-q\gamma^2)}+2a.
\end{equation}
Then there exists a stabilizing solution $\whR=\whR^\T\succcurlyeq
0$ to the algebraic Riccati equation
\begin{multline}\label{eq:aninorm Ric eq}
A^\T\whR A-\whR  +qC^\T C+(A^\T\whR B+qC^\T D)(I_m-B^\T\whR
B-qD^\T D)^{-1}(B^\T\whR A+qD^\T C)= 0
\end{multline}
such that
\begin{equation}\label{eq:Sigma positive whR}
I_m-B^\T \whR B-qD^\T D \succ 0
\end{equation}
and
\begin{equation}\label{eq:low ani ineq whR}
\ln\det{(I_m-B^\T\whR B-qD^\T D)}
> m\ln{(1-q\gamma^2)}+2a.
\end{equation}
Moreover, $\whR \prec R$.
\end{lemma}
\begin{proof}
Let us fix $q$. From~(\ref{eq:aninorm Ric ineq}) it follows that
there exists a real $(n\times n)$-matrix $Q=Q^\T\succ 0$ such that
\begin{equation}\label{eq:aninorm Ric ineq plus Q}
A^\T RA-R  + qC^\T C + Q +(A^\T RB+qC^\T D)(I_m-B^\T RB-qD^\T
D)^{-1}(B^\T RA+qD^\T C)= 0.
\end{equation}
Note that~(\ref{eq:Sigma positive}) also yields $I_m-qD^\T D\succ
0$. Then, by virtue of Lemma~2.1 in~\cite{DSX_1992} 
there exists a real $(n\times n)$-matrix $\whR=\whR^\T\succcurlyeq
0$ satisfying~(\ref{eq:aninorm Ric eq}) such that~(\ref{eq:Sigma
positive whR}) holds true and all eigenvalues of the matrix
$$
\ol{A} := A+B(I_m-B^\T\whR B-qD^\T D)^{-1}(B^\T\whR A+qD^\T C)
$$
lie within the closed unit disc. Furthermore, we have
\begin{equation}\label{eq:solutions ineq}
0\preccurlyeq\whR \preccurlyeq R.
\end{equation}
The inequalities (\ref{eq:low ani ineq}) and (\ref{eq:low ani ineq
whR}) can be rewritten as
\begin{eqnarray}\label{eq:trans low ani ineq}
\det{(I_m-B^\T R B-qD^\T D)} & > &
(1-q\gamma^2)^m\e^{2a},\\\label{eq:trans low ani ineq whR}
\det{(I_m-B^\T\whR B-qD^\T D)} & > & (1-q\gamma^2)^m\e^{2a},
\end{eqnarray}
respectively. From~(\ref{eq:solutions ineq})--(\ref{eq:trans low
ani ineq whR}) it can be seen that
$$
\det{(I_m-B^\T\whR B-qD^\T D)}\geqslant \det{(I_m-B^\T R B-qD^\T
D)}> (1-q\gamma^2)^m\e^{2a}
$$
which proves~(\ref{eq:low ani ineq whR}). Now, let us show that
the matrix $\ol{A}$ is actually stable, i.e. the matrix $\whR$ is
the stabilizing solution of the algebraic Riccati
equation~(\ref{eq:aninorm Ric eq}). Denoting $P:=-R$ and $\whP :=
-\whR$, the equations~(\ref{eq:aninorm Ric ineq plus Q}),
(\ref{eq:aninorm Ric eq}) can be rewritten as
$$
A^\T PA -P -qC^\T C-Q -(A^\T PB-qC^\T D)(I_m-qD^\T D+B^\T
PB)^{-1}(B^\T PA-qD^\T C)= 0,
$$
$$
A^\T\whP A -P-qC^\T C-(A^\T\whP B-qC^\T D)(I_m-qD^\T D+B^\T \whP
B)^{-1}(B^\T\whP A-qD^\T C)= 0,
$$
respectively. Applying Lemma~3.1 from~\cite{DS_1989} 
we have that the matrix $\whP-P$ must satisfy the following
equation:
\begin{equation}\label{eq:aninorm Ric in two matrices}
\whP-P = \olA^\T(\whP-P)\olA+\olA^\T(\whP-P)B(I_m-qD^\T D+B^\T
PB)^{-1}B^\T(\whP-P)\olA+Q.
\end{equation}
Suppose that the matrix $\olA$ is not stable, i.e. there exists a
nonzero vector $\zeta\in\R^n$ and scalar value $\lambda$,
$|\lambda|=1$, such that $\olA \zeta = \lambda \zeta$. Then
from~(\ref{eq:aninorm Ric in two matrices}) it follows that
\begin{equation}\label{eq:zero quadratic forms}
\zeta^\T\olA^\T(\whP-P)B(I_m-qD^\T D+B^\T PB)^{-1}B^\T(\whP-P)\olA
\zeta +\zeta^\T Q\zeta = 0.
\end{equation}
Since by~(\ref{eq:solutions ineq}) and (\ref{eq:Sigma positive})
\begin{multline*}
\zeta^\T\olA^\T(\whP-P)B(I_m-qD^\T D+B^\T PB)^{-1}B^\T(\whP-P)\olA
\zeta\\= \zeta^\T\olA^\T(R-\whR)B(I_m-qD^\T D-B^\T
RB)^{-1}B^\T(R-\whR)\olA \zeta\succcurlyeq 0
\end{multline*}
for all nonzero $\zeta$, from~(\ref{eq:zero quadratic forms}) it
follows that $\zeta^\T Q\zeta\leqslant 0$ for all nonzero $\zeta$.
This contradicts to the assumption that $Q\succ 0$. Therefore, the
matrix $\olA$ is stable, i.e. the matrix $\whR$ is the positive
semi-definite stabilizing solution to~(\ref{eq:aninorm Ric eq}).
Finally, from~(\ref{eq:aninorm Ric in two matrices}) it follows
that $\whR\prec R$, which completes the proof.
\end{proof}

\noindent\textit{Proof of Theorem~\ref{theorem:SANBRL det}.} Note
that by virtue of the Schur Theorem (see
e.g.~\cite{Bernstein_book_2005,Poznyak_2007}) the linear matrix
inequality~(\ref{eq:aninorm LMI 2x2}) is equivalent
to~(\ref{eq:aninorm Ric ineq}), (\ref{eq:Sigma positive}) for all
$q\in(0,\min(\gamma^{-2},$ $\|F\|_\infty^{-2}))$. The
inequality~(\ref{determinant inequality}) can be rewritten
as~(\ref{eq:low ani ineq}) and the strict form
of~(\ref{low_aniso}). By applying Lemma~\ref{lemma:from ineq to
eq}, we conclude that in this case there exists a stabilizing
solution to the Riccati equation~(\ref{eq:aninorm Ric eq}) such
that the inequality~(\ref{eq:low ani ineq whR}) holds true. Then,
by virtue of Theorem~1 in~\cite{KMT_2010} (see
Lemma~\ref{lemma:nonstrict ANBRL}), the
inequality~(\ref{eq:bounded aninorm ineq}) also holds, which was
to be proved.
\qquad\qquad\qquad\qquad\qquad\qquad\qquad\qquad\qquad\qquad\qquad\qquad\,\,\,\,$\square$
\begin{remark}
A solution to the inequalities~(\ref{determinant inequality}),
(\ref{eq:aninorm LMI 2x2}) of Theorem~\ref{theorem:SANBRL det} can
be found by means of available software packages for convex
optimization that allows the convex function
$-(\det(\cdot))^{1/m}$ to be used not only as an objective, but
also in constraints~\cite{Lofberg_2004}.
\end{remark}

\subsection{Computing anisotropic norm by convex
optimization}\label{sect:aninorm computing}

Being convex in both variables ${q}\in(0, \min(\gamma^{-2},
\|F\|_\infty^{-2}))$ and ${R}\succ0$, the
conditions~(\ref{determinant inequality}), (\ref{eq:aninorm LMI
2x2}) of Theorem~\ref{theorem:SANBRL det} are not directly
applicable for computing the minimal $\gamma$ such that the
inequality~(\ref{determinant inequality}) holds true because of
the product of $q$ and $\gamma^2$ on the right-hand side of the
inequality~(\ref{determinant inequality}). One of possible ways to
overcome this obstacle is to apply an auxiliary search algorithm
(for example, the interval bisection method) for finding the
minimal value of $\gamma$ such that the
inequalities~(\ref{determinant inequality}), (\ref{eq:aninorm LMI
2x2}) are solvable. This, however, would inevitably increase the
required computation time. Instead of doing so, let us multiply
both inequalities
$$
    -(\det{(I_m-B^\T RB-qD^\T D)})^{1/m}
    < -(1-q\gamma^2)\e^{2a/m},
$$
$$
\left[
\begin{array}{cc}
A^\T RA-R & A^\T RB\\
B^\T RA & B^\T RB-I_m
\end{array}
\right]+q\left[
\begin{array}{c}
C^\T\\D^\T
\end{array}
\right]\left[
\begin{array}{cc}
C & D
\end{array} \right] \prec 0
$$
of Theorem~\ref{theorem:SANBRL det} by $\eta:= q^{-1}>0$ recalling
that $q>0$ due to the strict localization in~(\ref{eq:q
localization inequalities}), see Remark~\ref{remark:nonstrict
ANBRL}. By rescaling the matrix $R$ as $\Phi := \eta R$, we can
make the SANBRL constraints linear in $\gamma^2$.
\begin{theorem}\label{theorem:SANBRL gamma2 linear}
Suppose the assumptions of Theorem~\ref{theorem:SANBRL det} are
satisfied. Then the $a$-anisotropic norm~(\ref{anorm}) of system
$F$ is strictly bounded by a given threshold $\gamma>0$, i.e.
$$
\sn F\sn_a < \gamma
$$
if there exists $\eta > \gamma^2$
 such that the inequality
\begin{equation}
    \label{eq:determinant inequality gamma2 linear}
    \eta-(\det({\e^{-2a/m}(\eta I_m-B^\T\Phi B-D^\T D)}))^{1/m}
    < \gamma^2
\end{equation}
holds true  for the real $(n\times n)$-matrix $\Phi=\Phi^\T\succ
0$ satisfying the linear matrix inequality
\begin{equation}\label{eq:aninorm LMI 2x2 in Phi eta}
\left[
\begin{array}{cc}
A^\T\Phi A -\Phi + C^\T C & A^\T\Phi B + C^\T D\\
B^\T\Phi A + D^\T C & B^\T\Phi B + D^\T D - \eta I_m
\end{array}
\right] \prec 0.
\end{equation}
\end{theorem}
\begin{remark}
With the notation $\wh{\gamma} := \gamma^2$, the conditions of
Theorem~\ref{theorem:SANBRL gamma2 linear} allow the minimal
$\gamma$ to be computed from a solution to the following convex
optimization problem:
\begin{equation}\label{eq:minimum gamma2}
\begin{array}{c}
\mathrm{find}\quad\wh{\gamma_\star} = \inf{\wh{\gamma}}\\
\mathrm{over}\quad \Phi, \eta, \wh{\gamma}
\quad\mathrm{satisfying}\quad\mbox{(\ref{eq:determinant inequality
gamma2 linear}), (\ref{eq:aninorm LMI 2x2 in Phi eta})}.
\end{array}
\end{equation}
Once the minimal $\wh{\gamma}_\star$ is found, the $a$-anisotropic
norm of the system $F$ is computed as
\begin{equation}\label{eq:aninorm computing}
\sn F\sn_a = \sqrt{\wh{\gamma}_\star}.
\end{equation}
\end{remark}
Note that, in contrast to the results
of~\cite{VKS_1996_1,DVKS_2001}, the presented technique for
computing the $a$-anisotropic norm does not employ the solution of
a complex system of cross-coupled equations via a homotopy-based
iterative algorithm~\cite{DKSV_1997_report}. In
Section~\ref{sect:numerical example} we will consider benchmark
results which demonstrate advantages and drawbacks of our convex
optimization approach in comparison with the earlier method.

\subsection{Limiting cases}\label{sect:limiting cases}

Let us now consider the conditions of Theorem~\ref{theorem:SANBRL
gamma2 linear} in two important cases when the mean anisotropy
level $a$ is equal to zero and tends to infinity, respectively.
Since the scaled $\H2$ norm and $\Hinf$ norm are two limiting
cases of the $a$-anisotropic norm as $a\to0,+\infty$
(see~(\ref{limits})), the inequalities~(\ref{eq:determinant
inequality gamma2 linear}), (\ref{eq:aninorm LMI 2x2 in Phi eta})
are expected to provide the criteria for verifying if the scaled
$\H2$ norm and $\Hinf$ norm of the system $F$ are bounded by a
given threshold $\gamma$.

First, we study the case of zero mean anisotropy level under the
convex constraints of Theorem~\ref{theorem:SANBRL gamma2 linear},
when the inequality~(\ref{eq:determinant inequality gamma2
linear}) becomes
\begin{equation}\label{eq:det inequality a=0}
   \eta -(\det{(\eta I_m-B^\T\Phi B-D^\T D)})^{1/m}
    < \gamma^2.
\end{equation}
By applying the arithmetic-geometric mean inequality to the
eigenvalues of the matrix $\eta I_m-B^\T\Phi B-D^\T D\succ0$, it
follows that
$$
(\det{(\eta I_m-B^\T\Phi B-D^\T D)})^{1/m}\leqslant
\ds\frac{1}{m}\tr(\eta I_m-B^\T\Phi B-D^\T D)
$$
(see e.g.~\cite[p.~275]{Bernstein_book_2005}.) So,
from~(\ref{eq:det inequality a=0}) it follows that
$$
\eta-\ds\frac{1}{m}\tr(\eta I_m-B^\T\Phi B-D^\T D) < \gamma^2
$$
or, equivalently,
\begin{equation}\label{eq:H2 trace inequality}
\tr(B^\T\Phi B+D^\T D) < m\gamma^2.
\end{equation}
By virtue of the Schur Theorem, the LMI~(\ref{eq:aninorm LMI 2x2
in Phi eta}) is equivalent to
$$
A^\T\Phi A-\Phi+C^\T C \prec (A^\T\Phi B + C^\T D)(B^\T\Phi B +
D^\T D - \eta I_m)^{-1}(A^\T\Phi B + C^\T D)^\T,
$$
$$
B^\T\Phi B + D^\T D - \eta I_m \prec 0,
$$
which implies that
\begin{equation}\label{eq:H2 Lyapunov inequality}
A^\T\Phi A-\Phi+C^\T C \prec 0.
\end{equation}
Now note that the fulfillment of the inequalities~(\ref{eq:H2
trace inequality}) and (\ref{eq:H2 Lyapunov inequality})
 is equivalent to
\begin{equation}\label{eq:H2 norm inequality}
\frac{1}{\sqrt{m}}\|F\|_2 < \gamma
\end{equation}
(see e.g.~\cite{Poznyak_2007}.)

In the case when $a\to+\infty$, the localization
$\gamma^2<\eta<\gamma^2/(1-\e^{-2a/m})$ yields
$\eta\to\gamma^{2}$; the inequality~(\ref{eq:determinant
inequality gamma2 linear}) becomes ineffective. In this case, by
rescaling the matrix $\bar\Phi:=\gamma\Phi$ and the Schur Theorem,
the LMI~(\ref{eq:aninorm LMI 2x2}) can be rewritten in the form
\begin{equation}\label{eq:Hinf LMI 3x3}
\left[
\begin{array}{ccc}
A^\T\bar\Phi A-\bar\Phi & A^\T \bar\Phi B & C^\T\\
B^\T \bar\Phi A & B^\T \bar\Phi B-\gamma I_m & D^\T\\
C & D & -\gamma I_p
\end{array}
\right] \prec 0
\end{equation}
which is well-known in the context of the discrete time $\Hinf$
control (see e.g.~\cite{DSX_1992,GA_1994}.) This fact is closely
related to the convergence $\lim_{a\to+\infty}{\sn F\sn_a} =
\|F\|_\infty$ in~(\ref{limits}) whereby the
inequality~(\ref{eq:bounded aninorm ineq}) `approximates'
\begin{equation}\label{eq:bounded Hinfnorm inquality}
\|F\|_\infty<\gamma
\end{equation}
for sufficiently large values of $a.$ Thus, in the limit, as
$a\to+\infty$, Theorem~\ref{theorem:SANBRL gamma2 linear} becomes
$\Hinf$ Bounded Real Lemma which establishes the equivalence
between~(\ref{eq:bounded Hinfnorm inquality}) and existence of a
positive definite solution to the LMI~(\ref{eq:Hinf LMI 3x3}).

\section{Numerical experiments and computational
benchmark}\label{sect:numerical example}

We have performed extensive numerical experiments to test the
efficiency and reliability of the proposed convex optimization
technique for computing the $a$-anisotropic norm of LDTI systems.
The computations, whose results are provided below, have been
carried out by means of MATLAB 7.9.0 (R2009b) and Control System
Toolbox in combination with the YALMIP
interface~\cite{Lofberg_2004} and SeDuMi solver~\cite{Sturm_1999}
with CPU P8700 $2\times2.53$GHz.

Let us first note that the number of variables of the resulting
convex optimization problem~(\ref{eq:determinant inequality gamma2
linear})--(\ref{eq:aninorm computing}) is $\frac{1}{2}n(n+1)+2$
and does not depend on the dimensions of the system input and
output, whereas the size of the LMI~(\ref{eq:aninorm LMI 2x2 in
Phi eta}) is $(n+m)\times(n+m)$ and does not depend on the system
output dimension $p$ either. The number of unknown variables in
the equation system of~\cite{VKS_1996_1,DVKS_2001} is $n(n+1)+1$.
For this reason, we carried out the computational experiments for
some fixed $p$. Using the MATLAB functions \texttt{drss} and
\texttt{randn}, we randomly generated 100 state-space realizations
of LDTI systems with random (positive) sampling time for each
combination of the dimensions from the sets $n = \{1\ldots12\}$,
$m = \{3,4,5\}$, $p = 2$. Thus, we obtained 3600 stable
realizations, possibly with poles arbitrarily close to the
boundary of the unit circle (up to the machine epsilon). For each
of them, we computed the $a$-anisotropic norm via the solution of
the convex optimization problem (COP) of Section~\ref{sect:aninorm
computing} and by I.G.\,Vladimirov's homotopy-based algorithm
(HBA)~\cite{DKSV_1997_report} for solving the system of three
cross-coupled nonlinear matrix algebraic equations derived
in~\cite{VKS_1996_1,DVKS_2001}. The computations were carried out
for 27 different values of the input mean anisotropy level
$a\in[0,20]$. Thus, the compared algorithms run 97200 times. The
required accuracy (tolerance) in all computations was set to
$10^{-9}$.

In computing the $a$-anisotropic norm by solving the convex
optimization problem we considered a run to be failed if the
optimization problem appeared to be infeasible or an unexpected
solver crash happened. If the issues were caused by the solver
itself, but the solution was, nevertheless, found, the run was
considered to be successful. In applying the homotopy-based
algorithm we stopped computations and concluded that the algorithm
fails if the prescribed accuracy had not been achieved after 2500
iterations. Also, a run of the homotopy-based algorithm was
considered to be a failure if one of the equations appeared to be
insolvable or an unexpected crashes of the MATLAB solvers for
Lyapunov and Riccati equations happened. Here, by the `solver
crashes' we mean those which do not originate from a particular
numerical algorithm used. Nevertheless, these events have also
been taken into consideration while assessing the reliability.

\begin{table}[!thpb]
  \caption{Mean CPU time required; $n=1\ldots 12,$ $m=3$, $p=2$}
  \label{table:test results time m=3 p=2}
  \begin{center}\scriptsize
  \begin{tabular}{|c||c|c|c||c|}\hline
    & \multicolumn{3}{|c||}{\textbf{COP}} & \textbf{HBA}\\\hline
$n$ & \textbf{Mean CPU} & \textbf{Mean YALMIP} & \textbf{Mean
SeDuMi} & \textbf{Mean
CPU}\\
 & \textbf{time (s)} & \textbf{time (s)} & \textbf{time (s)} & \textbf{time (s)}\\\hline
1 & 0.4840  & 0.2652 & 0.1161 & 0.4448\\
2 & 0.7944 & 0.4411 & 0.1406 & 0.5530\\
3 & 1.2102 & 0.6618 & 0.1690 & 1.0521\\
4 & 1.5484 & 0.8503 & 0.1722 & 0.9302\\
5 & 2.1429 & 1.1148 & 0.2851 & 1.2997\\
6 & 2.5697 & 1.4755 & 0.2555 & 1.7038\\
7 & 2.9299 & 1.6200 & 0.2245 & 1.4774\\
8 & 3.4697 & 1.8860 & 0.2418 & 1.6226\\
9 & 4.0750 & 2.1866 & 0.2515 & 1.8937\\
10 & 4.8381 & 2.5122 & 0.2794 & 1.9718\\
11 & 5.5680 & 2.9051 & 0.3054 & 2.0984\\
12 & 6.3453 & 3.2828 & 0.3387 & 2.8205\\\hline
  \end{tabular}
  \end{center}
\end{table}

\begin{table}[!thpb]
  \caption{Mean CPU time required; $n=1\ldots 12,$ $m=5$, $p=2$}
  \label{table:test results time m=5 p=2}
  \begin{center}\scriptsize
  \begin{tabular}{|c||c|c|c||c|}\hline
    & \multicolumn{3}{|c||}{\textbf{COP}} & \textbf{HBA}\\\hline
$n$ & \textbf{Mean CPU} & \textbf{Mean YALMIP} & \textbf{Mean
SeDuMi} & \textbf{Mean
CPU}\\
 & \textbf{time (s)} & \textbf{time (s)} & \textbf{time (s)} & \textbf{time (s)}\\\hline
1 & 0.6575 & 0.3234 & 0.1317 & 0.2111\\
2 & 1.1681 & 0.6147 & 0.1588 & 0.3328\\
3 & 1.6782 & 0.9088 & 0.1730 & 0.4330\\
4 & 2.2269 & 1.2423 & 0.1936 & 0.5451\\
5 & 2.8304 & 1.5783 & 0.2162 & 0.7714\\
6 & 3.4233 & 1.8830 & 0.2088 & 0.6867\\
7 & 4.0856 & 2.2377 & 0.2345 & 0.9555\\
8 & 5.1935 & 2.8440 & 0.2464 & 1.0044\\
9 & 6.0724 & 3.2426 & 0.2739 & 1.2394\\
10 & 7.0646 & 3.7505 & 0.2942 & 1.3387\\
11 & 7.9707 & 4.2034 & 0.3230 & 1.7716\\
12 & 8.9616 & 4.7629 & 0.3615 & 1.8914\\\hline
  \end{tabular}
  \end{center}
\end{table}

The benchmark results for $m = \{3,5\}$ are presented in
Tables~\ref{table:test results time m=3 p=2}--\ref{table:test
results fails a} and Figure~\ref{fig:CPU time}. The results for $m
= 4$ do not contradict the general tendency and, for the sake of
brevity,  are not presented here. In Tables~\ref{table:test
results time m=3 p=2}, \ref{table:test results time m=5 p=2}, the
mean CPU time required to compute the anisotropic norm was
calculated as the average value over all realizations of equal
dimensions and over a set of 27 different values of the input mean
anisotropy level $a\in[0,20]$. Comparison of the data shows that
computation of the $a$-anisotropic norm from the solution to COP
requires on average more CPU time than its computation by HBA.
Moreover, the average CPU time grows not only with the system
order $n$ but also with the system input dimension $m$ much faster
than that for HBA. Furthermore, the time required by the YALMIP
interface to form the optimization constraints is affected by the
number of these constraints which depends on the input dimension
$m$ and growth considerably with increase of $m$ in comparison
with the time required by the SeDuMi solver.

\begin{figure}[!thpb]
\begin{minipage}{81mm}
      \centering
      \includegraphics[width=81mm]{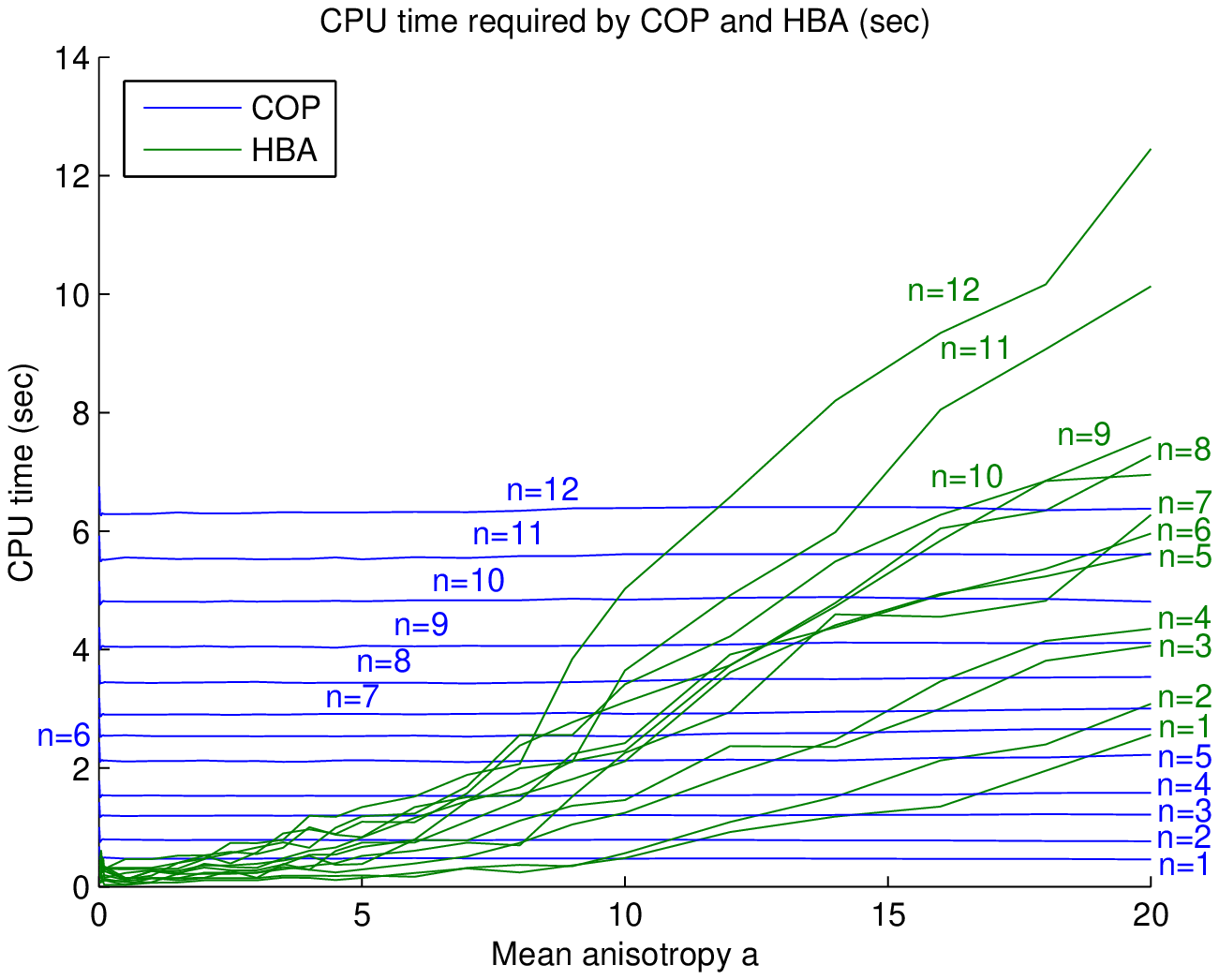}\newline
      (a)
\end{minipage}
\begin{minipage}{81mm}
      \centering
      \includegraphics[width=81mm]{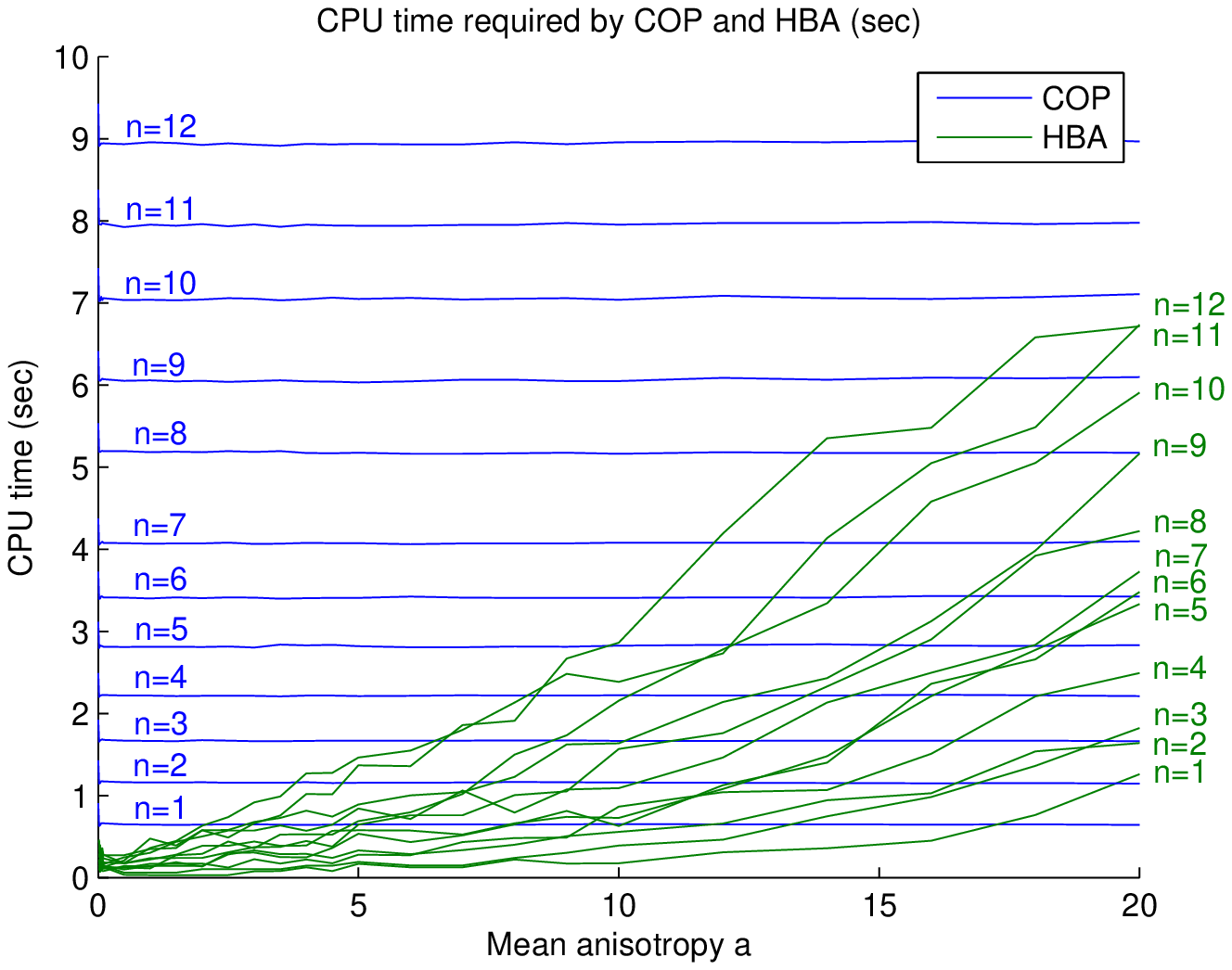}\newline
      (b)
      \end{minipage}
\caption{Mean CPU time required to compute the $a$-anisotropic
norm by the convex optimization (COP) and homotopy-based algorithm
(HBA); $n = \{1\ldots 12\}$, $p=2$, $m = 3$ (a), $m = 5$
(b).}\label{fig:CPU time}
\end{figure}

At the same time, the average values in Tables~\ref{table:test
results time m=3 p=2}, \ref{table:test results time m=5 p=2} do
not take into account the growth of the mean CPU time required by
HBA over all realizations of equal dimensions as the mean
anisotropy level $a$ increases. This growth is clearly
demonstrated by the diagrams in Figure~\ref{fig:CPU time}, where
the mean CPU time is shown as a function of the mean anisotropy
level $a$ for all groups of realizations of equal dimensions.
These diagrams also show that the mean CPU time required by COP
does not change noticeably with the increase of $a$.

The data in Tables~\ref{table:test results fails m=3 p=2},
\ref{table:test results fails m=5 p=2} are concerned with the
reliability of the algorithms being compared. The percentages of
successful and failed runs, infeasible problems, as well as runs
with numerical problems (COP) including the maximum admissible
number of iterations (HBA) exceeded were calculated as the average
value over all realizations of equal dimensions and over the set
of 27 different values of the input mean anisotropy level
$a\in[0,20]$. The analysis of Tables~\ref{table:test results fails
m=3 p=2}, \ref{table:test results fails m=5 p=2} shows that
computation of the $a$-anisotropic norm from the solution to COP
have more successful runs than HBA  on average. Moreover, all
failed runs of the optimization-based algorithm are caused by
infeasibility of the respective COP. Their fraction corresponds to
the percentage of realizations with poles located very close to
the unit circle in the total number of tested realizations. It
should be noted that HBA had the same percentage of runs failed
because of the infeasibility of algebraic Riccati equation.
However, this algorithm is also characterized by a certain
percentage of runs with the maximum number of iterations exceeded
and runs which resulted in unexpected crashes in the Lyapunov and
Riccati equation solvers.

\begin{table}[!thpb]
  \caption{Successful and failed runs; $n=1\ldots 12,$ $m=3$, $p=2$}
  \label{table:test results fails m=3 p=2}
  \begin{center}\scriptsize
  \begin{tabular}{|c||c|c|c|c||c|c|c|c|}\hline
  & \multicolumn{4}{|c||}{\textbf{COP}} & \multicolumn{4}{|c|}{\textbf{HBA}}\\\hline
 & \textbf{Succ.}  &\textbf{Failed}  & \textbf{Infeas.} & \textbf{Numer.} & \textbf{Succ.} & \textbf{Failed} &  \textbf{Infeas.} & \textbf{Max. iter.}\\
$n$ &   \textbf{(\%)} & \textbf{(\%)} & \textbf{(\%)} &
 \textbf{probl. (\%)} & \textbf{(\%)} & \textbf{(\%)} & \textbf{(\%)} &  \textbf{exceed. (\%)}\\\hline
1 & 100 & 0 & 0 & 5.1538 & 85.5385 & 14.4615 & 0 & 9.0385\\
2 & 99 & 1 & 1 & 4.0385 & 81.6923 & 18.3077 & 1 & 10.4231\\
3 & 90.1154 & 9.8846 & 9.8846 & 5.1154 & 70.3462 & 29.6538 & 9.8846 & 12.7308\\
4 & 95.5769 & 4.4231 & 4.4231 & 7.0385 & 75.2692 & 24.7308 & 4.4231 & 11.5769\\
5 & 92 & 8 & 8 & 9.2692 & 70.2308 & 29.7692 & 8 & 12.3462\\
6 & 91 & 9 & 9 & 12.2692 & 66.8846 & 33.1154 & 9 & 15.3846\\
7 & 94.7692 & 5.2308 & 5.2308 & 16.0769 & 68.0769 & 31.9231 & 5.2308 & 14.3077\\
8 & 88.1154 & 11.8846 & 11.8846 & 15.8077 & 67.6154 & 32.3846 & 11.8846 & 13.7692\\
9 & 92.3846 & 7.6154 & 7.6154 & 18.3846 & 63.5385 & 36.4615 & 7.6154 & 16.5000\\
10 & 88.4231 & 11.5769 & 11.5769 & 21.5000 & 62.8846 & 37.1154 & 11.5769 & 13.5000\\
11 & 89.8462 & 10.1538 & 10.1538 & 21.1154 & 65.7692 & 34.2308 & 10.1538 & 12.6923\\
12 & 91.4231 & 8.5769 & 8.5769 & 26.3077 & 66.4231 & 33.5769 & 8.5769 & 16\\
\hline
  \end{tabular}
  \end{center}
\end{table}

\begin{table}[!thpb]
  \caption{Successful and failed runs; $n=1\ldots 12,$ $m=5$, $p=2$}
  \label{table:test results fails m=5 p=2}
  \begin{center}\scriptsize
  \begin{tabular}{|c||c|c|c|c||c|c|c|c|}\hline
  & \multicolumn{4}{|c||}{\textbf{COP}} & \multicolumn{4}{|c|}{\textbf{HBA}}\\\hline
 & \textbf{Succ.}  &\textbf{Failed}  & \textbf{Infeas.} & \textbf{Numer.} & \textbf{Succ.} & \textbf{Failed} &  \textbf{Infeas.} & \textbf{Max. iter.}\\
$n$ &   \textbf{(\%)} & \textbf{(\%)} & \textbf{(\%)} &
 \textbf{probl. (\%)} & \textbf{(\%)} & \textbf{(\%)} & \textbf{(\%)} &  \textbf{exceed. (\%)}\\\hline
1 & 100 & 0 & 0 & 4.1154 & 92.0769 & 7.9231 & 0 & 3.6923\\
2 & 97 & 3 & 3 & 3.7308 & 87.1538 & 12.8462 & 3 & 5.0769\\
3 & 93 & 7 & 7 & 3.6923 & 82.1154 & 17.8846 & 7 & 6.0769\\
4 & 95.8462 & 4.1538 & 4.1538 & 4.4615 & 82.9231 & 17.0769 & 4.1538 & 7.9615\\
5 & 92.6923 & 7.3077 & 7.3077 & 4.6154 & 77.4615 & 22.5385 & 7.3077 & 9.6923\\
6 & 97.7308 & 2.2692 & 2.2692 & 4.8846 & 80.7308 & 19.2692 & 2.2692 & 8.2692\\
7 & 87.3846 & 12.6154 & 12.6154 & 5.1538 & 69.6154 & 30.3846 & 12.6154 & 10.6538\\
8 & 89.9615 & 10.0385 & 10.0385 & 4.5000 & 77.9615 & 22.0385 & 10.0385 & 9.9615\\
9 & 88.3462 & 11.6538 & 11.6538 & 5.4615 & 70.9231 & 29.0769 & 11.6538 & 10.5769\\
10 & 92.7308 & 7.2692 & 7.2692 & 5.9231 & 74.8462 & 25.1538 & 7.2692 & 12.2692\\
11 & 86.5769 & 13.4231 & 13.4231 & 6.9615 & 67.6538 & 32.3462 & 13.4231 & 12.6538\\
12 & 91.3077 & 8.6923 & 8.6923 & 7.9231 & 70.2692 & 29.7308 & 8.6923 & 14.3846\\
\hline
  \end{tabular}
  \end{center}
\end{table}

Finally, Table~\ref{table:test results fails a} gathers together
the mean CPU time required and percentages of successful and
failed runs computed as average values over all realizations
irrespective of dimensions for different values of the input mean
anisotropy level $a\in[0,20]$. It can be seen that the mean CPU
time required by HBA grows with increase of $a$. The same is true
in regard to the percentage of the HBA runs failed by the maximum
number of iterations exceeded. The percentage of HBA successful
runs decreases considerably as $a$ increases. At the same time,
the mean CPU time required by COP and the percentage of successful
runs of this algorithm change insignificantly with the growth of
the input mean anisotropy level.

\begin{table}[!thpb]
  \caption{Mean CPU time required, successful and failed runs for different values of $a$; all tested realizations}
  \label{table:test results fails a}
  \begin{center}\scriptsize
  \begin{tabular}{|l||c|c|c|c||c|c|c|c|}\hline
  & \multicolumn{4}{|c||}{\textbf{COP}} & \multicolumn{4}{|c|}{\textbf{HBA}}\\\hline
 & \textbf{Mean CPU} & \textbf{Succ.}  & \textbf{Infeas.} & \textbf{Numer.} & \textbf{Mean CPU} & \textbf{Succ.} &  \textbf{Infeas.} & \textbf{Max. iter.}\\
$a$ & \textbf{time (s)} & \textbf{(\%)} & \textbf{(\%)} &
 \textbf{probl. (\%)} & \textbf{time (s)} & \textbf{(\%)} & \textbf{(\%)} &   \textbf{exceed. (\%)}\\\hline
0 & 3.9341 & 93.2083 & 6.7917 & 14.5417 & --- & 0 & --- & 0\\
0.02 & 3.6000 & 93.2083 & 6.7917 & 6.7917 & 0.2864 & 88.7083 & 6.7917 &  2.3333\\
0.04 & 3.6098 & 93.2083 & 6.7917 & 6.7917 & 0.2400 & 89.8333 & 6.7917 &  1.7500\\
0.06 & 3.6273 & 93.1667 & 6.8333 & 6.8333 & 0.2261 & 90.4167 & 6.8333 &  1.6250\\
0.08 & 3.6282 & 93.1667 & 6.8333 & 6.8333 & 0.2113 & 90.7500 & 6.8333 &  1.4167\\
0.1 & 3.6246 & 93.1667 & 6.8333 & 6.8333 & 0.1893 & 91.0417 & 6.8333 &  1.1250\\
0.5 & 3.6184 & 93.1667 & 6.8333 & 6.8750 & 0.1615 & 92.2083 & 6.8333 &  0.6667\\
1 & 3.6185 & 93.1667 & 6.8333 & 7.0417 & 0.2184 & 91.3333 & 6.8333 &  1.5833\\
1.5 & 3.6175 & 93.0417 & 6.9583 & 7.1667 & 0.2509 & 90.9167 & 6.9583 &  2.0000\\
2 & 3.6189 & 93.0000 & 7.0000 & 7.2500 & 0.3209 & 89.8333 & 7.0000 &  3.0417\\
2.5 & 3.6195 & 92.9167 & 7.0833 & 7.3750 & 0.3574 & 89.3750 & 7.0833 &  3.4583\\
3 & 3.6179 & 92.8750 & 7.1250 & 7.4167 & 0.3926 & 88.7917 & 7.1250 &  3.9167\\
3.5 & 3.6163 & 92.8333 & 7.1667 & 7.4583 & 0.4593 & 87.8750 & 7.1667 &  4.7917\\
4 & 3.6196 & 92.7917 & 7.2083 & 7.5417 & 0.5338 & 86.7500 & 7.2083 &  5.7917\\
4.5 & 3.6206 & 92.7917 & 7.2083 & 7.6250 & 0.5498 & 86.3750 & 7.2083 &  6.0000\\
5 & 3.6197 & 92.7083 & 7.2917 & 7.7083 & 0.6754 & 84.3333 & 7.2917 &  7.6667\\
6 & 3.6213 & 92.5000 & 7.5000 & 8.0833 & 0.7531 & 83.1250 & 7.5000 &  8.4167\\
7 & 3.6201 & 92.3750 & 7.6250 & 8.3750 & 0.9554 & 79.9167 & 7.6250 &  10.7500\\
8 & 3.6265 & 92.2083 & 7.7917 & 8.8333 & 1.1690 & 76.5417 & 7.7917 &  13.0417\\
9 & 3.6308 & 92.1250 & 7.8750 & 9.3333 & 1.4899 & 72.0000 & 7.8750 &  16.3333\\
10 & 3.6307 & 92.1250 & 7.8750 & 9.8750 & 1.7947 & 68.0417 & 7.8750 &  19.7083\\
12 & 3.6445 & 92.1667 & 7.8333 & 10.6250 & 2.4873 & 58.7500 & 7.8333 &  27.4167\\
14 & 3.6441 & 92.2917 & 7.7083 & 11.2917 & 3.1595 & 49.3750 & 7.7083 &  33.6250\\
16 & 3.6492 & 92.1250 & 7.8750 & 12.2083 & 3.8370 & 40.5000 & 7.8750 &  37.5833\\
18 & 3.6517 & 92.0417 & 7.9583 & 12.7917 & 4.4236 & 32.8750 & 7.9583 &  39.6667\\
20 & 3.6545 & 92.2917 & 7.7083 & 12.9167 & 5.1172 & 26.5000 & 7.7083 &  38.2917\\
\hline
  \end{tabular}
  \end{center}
\end{table}

\section{Conclusion}\label{sect:conclusion}

We have introduced the Strict Anisotropic Norm Bounded Real Lemma
in terms of inequalities providing a state-space criterion for
verifying if the anisotropic norm of a LDTI system is bounded by a
given threshold value. This result extends the $\Hinf$ Bounded
Real Lemma to stochastic systems where the statistical
uncertainty, present in the random disturbances, is quantified by
the mean anisotropy level.

The derived criterion employs the solution of an LMI and an
inequality on the determinant of a related positive definite
matrix and a positive scalar parameter. SANBRL in terms of
inequalities provides a key result which is used for the design of
suboptimal (or $\gamma$-optimal) anisotropic controllers via
convex optimization and semidefinite programming to ensure a
specified upper bound on the anisotropic norm of the closed-loop
system (respectively, to minimize the norm). It can also be
combined with additional specifications for the controllers.

\section*{Acknowledgements}

The authors deeply appreciate the invaluable help of Igor
Vladimirov who corrected the statement and proof of the main
result in~\cite{KMT_2010} and actually rearranged the text of that
paper. The authors kindly thank Didier Henrion and Arkadii
Nemirovskii for helpful discussions and advices on the formulation
of Theorems~\ref{theorem:SANBRL det} and \ref{theorem:SANBRL
gamma2 linear}. Useful comments of the anonymous reviewer are also
gratefully appreciated.

This work was supported by the Russian Foundation for Basic
Research (grant 11-08-00714-a) and Program for Fundamental
Research No.~15 of EEMCP Division of Russian Academy of Sciences.

\end{document}